\documentclass{elsarticle}

\usepackage{pgf}
\usepackage{tikz}
\usepackage{comment}
\usepackage{url}
\usepackage{amsthm}
\usepackage{graphicx}
\usepackage{amsmath}
\usepackage{multirow}
\usepackage{rotating}
\usepackage[caption=false]{subfig}

\newcommand{\Oh}[1]{\ensuremath{\mathcal{O}\!\left( {#1} \right)}}
\newcommand{\PL}{\ensuremath{\mathsf{PL}}}

\newtheorem{thrm}{Theorem}
\newtheorem{lemm}[thrm]{Lemma}

\newcommand{\GPL}{\ensuremath{\mathsf{GPL}}}
\newcommand{\rev}[1]{{#1}^R}

\def\balgorithm#1{{\bf Algorithm \texttt{#1}}}

\def\bfor{{\bf for\ }}
\def\bforeach{{\bf foreach\ }}
\def\bto{{\bf to\ }}

\def\band{{\bf and\ }}

\def\bdo{{\bf do\ }}
\def\bif{{\bf if\ }}
\def\bthen{{\bf then\ }}
\def\belse{{\bf else\ }}

\def\breturn{{\bf return\ }}

\def\la{\leftarrow}

\def\+{\!+\!}
\def\-{\!-\!}

\begin{document}

\begin{frontmatter}

\journal{Journal of Discrete Algorithms}

\title{A Subquadratic Algorithm for \mbox{Minimum Palindromic Factorization}}

\author[up]{Gabriele Fici}\ead{gabriele.fici@unipa.it}
\author[uh]{Travis Gagie}\ead{travis.gagie@cs.helsinki.fi}
\author[uh]{Juha K\"arkk\"ainen}\ead{juha.karkkainen@cs.helsinki.fi}
\author[uh]{Dominik Kempa}\ead{dominik.kempa@cs.helsinki.fi}

\address[up]{Dipartimento di Matematica e Informatica, Universit\`a di Palermo,
Italy}
\address[uh]{Department of Computer Science, University of Helsinki, Finland}


\begin{abstract}
  We give an $\Oh{n \log n}$-time, $\Oh{n}$-space algorithm for
  factoring a string into the minimum number of palindromic
  substrings. That is, given a string \(S [1..n]\), in $\Oh{n \log n}$
  time our algorithm returns the minimum number of palindromes \(S_1,
  \ldots, S_\ell\) such that \(S = S_1 \cdots S_\ell\).  We also show
  that the time complexity is $\Oh{n}$ on average and $\Omega(n\log
  n)$ in the worst case. The last result is based on a
  characterization of the palindromic structure of Zimin words.
\end{abstract}

\end{frontmatter}

\section{Introduction} \label{sec:introduction}

Palindromic substrings are a well-studied topic in stringology and
combinatorics on words.  Since a single character is a palindrome,
there are always between $n$ and \(\binom{n}{2} + n = \Theta (n^2)\)
non-empty palindromic substrings in a string of length $n$.  There are
only \(2 n - 1\) possible centers of those substrings,
however\;---\;i.e., the $n$ individual characters and the \(n - 1\)
gaps between them\;---\;so many algorithms involving palindromic
substrings still run in subquadratic time.  For example,
Manacher~\cite{Man75} gave a linear-time algorithm for listing all the
palindromic prefixes of a string.  Apostolico, Breslauer and
Galil~\cite{ABG95} observed that Manacher's algorithm can be used to
list in linear time all maximal palindromic substrings, which are
those that cannot be extended without changing the position of the
center.  Other linear-time algorithms for this problem were given by
Jeuring~\cite{Jeu94} and Gusfield~\cite{Gus97}.  Since any palindromic
substring is contained within the maximal palindromic substring with
the same center, the list of all maximal palindromic substrings can be
viewed as a linear-space representation of all palindromic substrings.
For more discussion of algorithms involving palindromes, we refer the
reader to Jeuring's recent survey~\cite{Jeu13}.

Palindromes are a useful tool for investigating string complexity; see, e.g.,~\cite{PalCom}.
A natural measure of the asymmetry of a string $S$ is its palindromic length \(\PL (S)\),
which is the minimum number of palindromic
substrings into which $S$ can be factored.  That is, \(\PL (S)\) is
the minimum number $\ell$ such that there exist palindromes \(S_1,
\ldots, S_\ell\) whose concatenation \(S_1 \cdots S_\ell = S\).  For
example, $\PL (abaab) = 2$ and $\PL (abaca) = 3$.  Notice that, since
a single character is a palindrome, \(\PL (S)\) is always well-defined
and lies between 0 and $|S|$, or 1 and $|S|$ if $S$ is non-empty.
In fact, \(\PL (S [1..i]) - 1 \leq \PL (S [1..i + 1]) \leq \PL (S [1..i]) + 1\) for \(i < |S|\):
first, if \(S_1, \ldots, S_{\ell - 1}, S [h..i + 1]\) is a factorization of \(S [1..i + 1]\) into $\ell$ palindromic substrings, then \(S_1, \ldots, S_{\ell - 1}, S [h], S [h + 1..i]\) is a factorization of \(S [1..i]\) into \(\ell + 1\) palindromic substrings;
second, if \(S_1, \ldots, S_\ell\) is a factorization of \(S [1..i]\) into $\ell$ palindromic substrings, then \(S_1, \ldots, S_\ell, S [i + 1]\) is a factorization of \(S [1..i + 1]\) into \(\ell + 1\) palindromic substrings.

We became interested in palindromic length because of a recent
conjecture by Frid, Puzynina and Zamboni~\cite{FPZ13}.
Some infinite strings (e.g., the regular paperfolding sequence) are
highly asymmetric in that they contain only a finite number of
distinct palindromic substrings; see~\cite{FiZa13} for more
discussion.  For such strings, the palindromic length of any finite
substring is proportional to that substring's length.  In contrast,
for other infinite strings (e.g., the infinite power of any
palindrome), the palindromic length of any finite substring is
bounded.  Frid et al.\ conjectured that all such infinite strings are
(ultimately) periodic.

It is easy to compute \(\PL (S)\) in quadratic time via dynamic
programming.  Alatabbi, Iliopoulos and Rahman~\cite{AIR13}
recently gave a linear-time algorithm for computing a minimum
factorization of $S$ into \emph{maximal} palindromic substrings, when
such a factorization exists; it does not exist for, e.g., \(abaca\).
Even when such a factorization exists, it may consist of more than
\(\PL (S)\) substrings; e.g., $abbaabaabbba$ can be factored into
\(abba\), \(aba\) and \(abbba\) but cannot be factored into fewer than
four maximal palindromic substrings.  

In this paper, we give an $\Oh{n \log n}$-time and $\Oh{n}$-space
algorithm for factoring $S$ into \(\PL (S)\) palindromic substrings.
The average case time complexity is in fact linear, but 
the worst case is $\Theta(n\log n)$, which we show by an
analysis of the palindromic structure of Zimin
words~\cite[Chapter~5.4]{BLRS08}.

Independently of us, I, Sugimoto, Inenaga, Bannai and
Takeda~\cite{ISIBT14} discovered essentially the same
algorithm. Also, Kosolobov, Rubinchik and
Shur~\cite{KRS14} have recently described an
algorithm recognizing strings with a given palindromic length. Their
result can be used for computing the palindromic length of a string $S$
in $\Oh{|S| \cdot \PL(S)}$ time.\footnote{Editors' note:
we are satisfied that the results of this paper, and
those of~\cite{ISIBT14} and~\cite{KRS14}, have all been achieved independently.}

\section{A Simple Quadratic Algorithm}
\label{sec:quadratic}

We start by describing a simple algorithm for computing \(\PL (S)\) in
$\Oh{n^2}$ time and $\Oh{n}$ space
using the observation that, for \(1 \leq j \leq n\),
\begin{equation*}
\label{eq:pl-formula}
  \PL (S [1..j]) = \min_i \left\{ \rule{0ex}{2ex} \PL (S [1..i - 1]) + 1\ : \ i
\leq j,\ \mbox{\(S [i..j]\) is a palindrome} \right\}\,.
\end{equation*}
We compute and store an array $\PL[0..n]$, where $\PL[0]=0$ and 
$\PL[i]=\PL(S[1..i])$ for $i \ge 1$. At each step $j$, we
compute the set $P_j$ of the starting positions of all
palindromes ending at $j$ from the set $P_{j-1}$ using the
observation that \(S [i..j]\), $i+1\le j-1$, is a palindrome if and only if
\(S [i + 1..j - 1]\) is a
palindrome and \(S [i] = S [j]\).  The algorithm is given in
Figure~\ref{fig-quadratic-algorithm}.
\begin{figure}[t]
\centering
\begin{tabbing}
00: \=\qquad\=\qquad\=\qquad\=\qquad\=\qquad\=\kill
 \balgorithm{Palindromic-length}($S[1..n]$)\\
 1:\>$\PL[0] \la 0$\\
 2:\>$P \la \emptyset$ \\
 3:\>\bfor $j \la 1$ \bto $n$ \bdo\\
 4:\>\>$P' \la \emptyset$\\
 5:\>\>\bforeach $i\in P$ \bdo\\
 6:\>\>\>\bif $i>1$ \band $S[i-1]=S[j]$ \bthen\\
 7:\>\>\>\>$P' \la P' \cup \{i-1\}$\\
 8:\>\>\bif $j>1$ \band $S[j-1]=S[j]$ \bthen\\
 9:\>\>\>$P' \la P' \cup \{j-1\}$\\
 10:\>\>$P \la P' \cup \{j\}$\\
 11:\>\>$\PL[j] \la j$\\
 12:\>\>\bforeach $i\in P$ \bdo\\
 13:\>\>\>$\PL[j] \la \min(\PL[j],\PL[i-1]+1)$\\
 14:\>\breturn $\PL[n]$
\end{tabbing}
\caption{A simple quadratic-time algorithm for computing the palindromic length.
Every iteration of the for loop in line 3 starts with $P=P_{j-1}$ and ends with
$P=P_j$.}
\label{fig-quadratic-algorithm}
\end{figure}

The space requirement is clearly $\Oh{n}$.  During the $j$th step of
the algorithm, we use time $\Oh{|P_j|+|P_{j-1}|}$, so for all the
steps we use total time proportional to the number of palindromic
substrings in $S$. For most strings the time is linear (see
Theorem~\ref{thrm:average}) but the worst case is quadratic, e.g., for
$S=a^n$ or \(S=(ab)^{n / 2}\).

It is straightforward to modify the algorithm so that it produces an
actual minimum palindromic factorization of $S$, without increasing
the running time or space by more than a constant factor.

\section{Faster Computation of Palindromes}
\label{sec:groups}

In this section, we replace the representation $P_j$ of 
the palindromes ending at $j$ with a more compact representation $G_j$
that needs only $\Oh{\log j}$ space
and can be computed in $\Oh{\log j}$
time from $G_{j-1}$. The
representation is based on combinatorial properties of palindromes.

A string $y$ is a \emph{border} of a string $x$ if $y$ is both a
prefix of $x$ and a suffix of $x$, and a \emph{proper} border if
$y\neq x$.  
The following easy
lemmas establish a connection between borders and palindromes.

\begin{lemm}[\cite{blondin2008palindromic}]
\label{lemm:palbor}
Let $y$ be a suffix of a palindrome $x$. Then $y$ is a border of $x$
iff $y$ is a palindrome.
\end{lemm}

\begin{lemm}[\cite{blondin2008palindromic}]
\label{lemm:pallongbor}
Let $x$ be a string with a border $y$ such that $|x|\le 2|y|$. Then
$x$ is a palindrome iff $y$ is a palindrome.
\end{lemm}

A positive integer $p\le |x|$ is a \emph{period} of a string $x$ 
if there exists a string $w$ of length $p$ such that $x$ is a factor
of $w^\infty$. It is well known that $y$ is a proper border of $x$ if
and only if $|x|-|y|$ is a period of $x$. This, together with
Lemma~\ref{lemm:palbor}, implies the following connection between periods
and palindromes.

\begin{lemm}
\label{lemm:palper}
Let $y$ be a proper suffix of a palindrome $x$. Then $|x|-|y|$ is a
period of $x$ iff $y$ is a palindrome. In particular, $|x|-|y|$ is the
smallest period of $x$ iff $y$ is the longest palindromic proper suffix
of $x$.
\end{lemm}

Now we are ready to state and prove the key combinatorial property of
palindromic suffixes.

\begin{lemm}
\label{lemm:gap-properties}
Let $x$ be a palindrome, $y$ the longest palindromic proper suffix of
$x$ and $z$ the longest palindromic proper suffix of $y$.
Let $u$ and $v$ be strings such that $x=uy$ and $y=vz$.
Then
\begin{enumerate}
\item[\rm(1)] $|u|\ge |v|$;
\item[\rm(2)] if $|u|>|v|$ then $|u|>|z|$;
\item[\rm(3)] if $|u|=|v|$ then $u=v$.
\end{enumerate}

\end{lemm}

\begin{figure}[t]
  \centering\small
  \begin{tikzpicture}[scale=0.40]
    \draw[dashed, fill=gray!10] (11,-2) rectangle (14,-1);
    \draw[dashed, fill=gray!10] (0,-1) rectangle (11,0);

    \draw[fill=gray!10] (0,0) rectangle (20,1);
    \draw[fill=gray!10] (11,-1) rectangle (20,0);
    \draw[fill=gray!10] (14,-2) rectangle (20,-1);

    \draw (10, 0.5) node {$x$} rectangle (10, 0.5);
    \draw (15.5, -0.5) node {$y$} rectangle (15.5, -0.5);
    \draw (5.5, -0.5) node {$u$} rectangle (5.5, -0.5);
    \draw (17, -1.5) node {$z$} rectangle (17, -1.5);
    \draw (12.5, -1.5) node {$v$} rectangle (12.5, -1.5);

    \draw[-] (0, 1) .. controls (6, 4) and (14, 4) .. (20, 1);
    \draw[-] (11, 1) .. controls (13.5, 2.5) and (17.5, 2.5) .. (20, 1);
    \draw[-] (14, 1) .. controls (15.5, 2) and (18.5, 2) .. (20, 1);
    
    \draw[-, dotted] (11, 0) -- (11, 1);
    \draw[-, dotted] (14, -1) -- (14, 1);
  \end{tikzpicture}
  \caption{Proof of Lemma~\ref{lemm:gap-properties}:
  $|u|\ge |v|$;
  if $|u|>|v|$ then $|u|>|z|$; and
  if $|u|=|v|$ then $u=v$.}
\label{fig:lemm-gap-properties}
\end{figure}
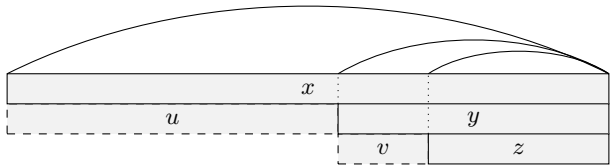

\begin{proof}

See Figure~\ref{fig:lemm-gap-properties} for
an illustration.

  (1) By Lemma~\ref{lemm:palper}, $|u|=|x|-|y|$ is the smallest period
  of $x$, and $|v|=|y|-|z|$ is the smallest period of $y$. Since $y$
  is a factor of $x$, either $|u|>|y|>|v|$ or $|u|$ is a period of $y$
  too, and thus it cannot be smaller than $|v|$.

(2) By Lemma~\ref{lemm:palbor}, $y$ is a border of $x$ and thus $v$ is
a prefix of $x$. Let $w$ be a string such that $x=vw$. Then $z$ is a
border of $w$ and $|w|=|zu|$, see Figure~\ref{fig:lemm-gap-properties2}.
Since we assume
$|u|>|v|$, we must have $|w|>|y|$. Suppose to the contrary that
$|u|\le |z|$. Then $|w|=|zu|\le 2|z|$, and by
Lemma~\ref{lemm:pallongbor}, $w$ is a palindrome. But this contradicts
$y$ being the longest palindromic proper suffix of $x$.

(3) In the proof of (2) we saw that $v$ is a prefix of $x$, and so is $u$ by
definition. Thus $u=v$ if $|u|=|v|$.
\end{proof}

\begin{figure}[t]
  \centering\small
  \begin{tikzpicture}[scale=0.40]
    \draw[dashed, fill=gray!10] (6,-2) rectangle (9,-1);
    \draw[dashed, fill=gray!10] (0,2) rectangle (3,3);
    \draw[dashed, fill=gray!10] (0,-1) rectangle (6,0);

    \draw[fill=gray!10] (0,0) rectangle (20,1);
    \draw[fill=gray!10] (0,1) rectangle (14,2);
    \draw[fill=gray!10] (3,2) rectangle (14,3);
    \draw[fill=gray!10] (3,3) rectangle (20,4);
    \draw[fill=gray!10] (6,-1) rectangle (20,0);
    \draw[fill=gray!10] (9,-2) rectangle (20,-1);

    \draw (10,0.5) node {$x$} rectangle (10,0.5);
    \draw (7,1.5) node {$y$} rectangle (7,1.5);
    \draw (8.5,2.5) node {$z$} rectangle (8.5,2.5);
    \draw (1.5,2.5) node {$v$} rectangle (1.5,2.5);
    \draw (11.5,3.5) node {$w$} rectangle (11.5,3.5);
    \draw (13,-0.5) node {$y$} rectangle (13,-0.5);
    \draw (3,-0.5) node {$u$} rectangle (3,-0.5);
    \draw (14.5,-1.5) node {$z$} rectangle (14.5,-1.5);
    \draw (7.5,-1.5) node {$v$} rectangle (7.5,-1.5);

    \draw[<->] (14,1.5) -- (20,1.5);
    \draw (17, 1.4) node[above] {$|u|$} rectangle (17, 1.4);
  \end{tikzpicture}
  \caption{Proof of Lemma~\ref{lemm:gap-properties}(2): if $|u|>|v|$ and
    $|u|\leq|z|$ then $w$ is a palindromic proper suffix of $x$ longer than $y$.}
\label{fig:lemm-gap-properties2}
\end{figure}
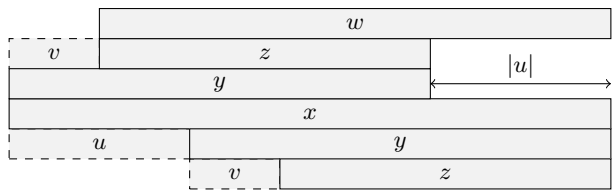

We will use the above lemma to establish the properties of the set $P_j$.
Let $P_j=\{p_1,p_2,\dots,p_m\}$ with $p_1<p_2<\dots<p_m$.  By
\textit{gap} we mean the difference $p_i-p_{i-1}$ of two consecutive
values in $P_j$.  The following result has been proven
in~\cite{MIISNH09} but we provide a proof for completeness.

\begin{lemm}
\label{corr:number-of-gaps}
The sequence of gaps in $P_j$ is non-increasing and there are
at most $\Oh{\log j}$ distinct gaps.
\end{lemm}

\begin{proof}
For any $i\in[2..m-1]$, if we let
$x=S[p_{i-1}..j]$, $y=S[p_i..j]$ and $z=S[p_{i+1}..j]$, we have the
situation of Lemma~\ref{lemm:gap-properties} with gaps of $|u|$ and
$|v|$. 
The sequence of gaps is
non-increasing by Lemma~\ref{lemm:gap-properties}(1).
If we have a change of gap, i.e., $|u|>|v|$, 
we must have $|x| > |u|+|z| > 2|z|$ by
Lemma~\ref{lemm:gap-properties}(2), i.e., the length of the palindromic
suffix is halved in two steps. This cannot happen more than $\Oh{\log j}$
times. 
\end{proof}

We will partition the set $P_j$ by the gaps into $\Oh{\log j}$
consecutive subsets, each of which can be represented in constant
space since it forms an arithmetic progression. 
For any positive integer $\Delta$, we
define $P_{j,\Delta} = \{p_i : 1 < i \le m, p_i-p_{i-1}=\Delta\}$, and
$P_{j,\infty}=\{p_1\}$. Each non-empty $P_{j,\Delta}$ is represented by the
triple $(\min P_{j,\Delta}, \Delta, |P_{j,\Delta}|)$. Let $G_j$ be the
list of such triples in decreasing order of $\Delta$.

The list $G_j$ is a full representation of $P_j$ of size $\Oh{\log
  j}$. We will show that $G_j$ can be computed from
$G_{j-1}$ in $\Oh{|G_{j-1}|}$ time. In the quadratic-time algorithm,
each element $i$ of $P_{j-1}$ was either eliminated or replaced by $i-1$
in $P_j$. The following lemma shows that the decision to
eliminate or replace can be made simultaneously for all elements of a
partition $P_{j-1,\Delta}$. See Figure~\ref{fig:G-repr-example} for
an example.

\begin{figure}[t]
  \centering
  \subfloat[]{
    \hspace{-1.07cm}
    \small
    \begin{tikzpicture}[scale=0.49]
    \foreach \x/\ch in { 1/c, 2/a, 3/a, 4/a, 5/b, 6/a, 7/a, 8/a, 9/b, 10/a, 11/a, 12/a, 13/b, 14/a, 15/a, 16/a}
      \draw (\x-1,0) rectangle (\x,1);
    \foreach \x/\ch in { 1/c, 5/b, 9/b, 13/b, 14/a, 15/a}
      \draw[fill=gray!25] (\x-1,0) rectangle (\x,1);
    \foreach \x/\ch in { 1/c, 2/a, 3/a, 4/a, 5/b, 6/a, 7/a, 8/a, 9/b, 10/a, 11/a, 12/a, 13/b, 14/a, 15/a, 16/a}
      \draw (\x-0.5,0) node[above] {\normalsize {\ch}};
    \foreach \x in { 1, 2, 3, 4, 5, 6, 7, 8, 9, 10, 11, 12, 13, 14, 15, 16}
      \draw (\x-0.5,-0.75) node[above] {\scriptsize \textcolor{gray}{\x}};

    \draw[-] (1, 1) .. controls (5,4.2) and (12,4.2) .. (16, 1);
    \draw[-] (5, 1) .. controls (8,3.2) and (13,3.2) .. (16, 1);
    \draw[-] (9, 1) .. controls (10.5,2.4) and (14.5,2.4) .. (16, 1);
    \draw[-] (13, 1) .. controls (14, 1.8) and (15,1.8) .. (16, 1);
    \draw[-] (14, 1) .. controls (14.7, 1.5) and (15.3,1.5) .. (16, 1);
    \draw[-] (15, 1) .. controls (15.3, 1.3) and (15.7,1.3) .. (16, 1);

    \draw (-2.0, 0.5) node {$S[1..j-1]:$} (-2.0, 0.5);
    \draw (-1.4, -1.5) node {$G_{j-1}:$} (-1.4, -1.5);

    \draw[rounded corners] (1,-2) rectangle (2,-1.1);
    \draw[rounded corners] (5,-2) rectangle (13.95,-1.1);
    \draw[rounded corners] (14.05,-2) rectangle (16,-1.1);

    \foreach \x in { 2, 6, 10, 14, 15, 16} \draw (\x-0.5,-2) node[above] {\x};
    \draw (1.5, -3) node[above] {$P_{j-1,\infty}$} (1.5, -3);
    \draw (9.5, -3) node[above] {$P_{j-1,4}$} (9.5, -3);
    \draw (15, -3) node[above] {$P_{j-1,1}$} (15, -3);
  \end{tikzpicture}
  \label{fig:G-repr-example}
  }

  \subfloat[]{
    \small
    \begin{tikzpicture}[scale=0.49]
      \foreach \x/\ch in { 1/c, 2/a, 3/a, 4/a, 5/b, 6/a, 7/a, 8/a, 9/b, 10/a, 11/a, 12/a, 13/b, 14/a, 15/a, 16/a, 17/b}
        \draw (\x-1,0) rectangle (\x,1);

      \foreach \x/\ch in { 5/b, 9/b, 13/b }
        \draw[fill=gray!25] (\x-1,0) rectangle (\x,1);

      \foreach \x/\ch in { 1/c, 2/a, 3/a, 4/a, 5/b, 6/a, 7/a, 8/a, 9/b, 10/a, 11/a, 12/a, 13/b, 14/a, 15/a, 16/a, 17/b}
        \draw (\x-0.5,0) node[above] {\normalsize {\ch}};
      \foreach \x in { 1, 2, 3, 4, 5, 6, 7, 8, 9, 10, 11, 12, 13, 14, 15, 16, 17}
        \draw (\x-0.5,-0.75) node[above] {\scriptsize \textcolor{gray}{\x}};

      \draw[-] (4, 1) .. controls (8,3.2) and (13,3.2) .. (17, 1);
      \draw[-] (8, 1) .. controls (10,2.4) and (15,2.4) .. (17, 1);
      \draw[-] (12, 1) .. controls (13.5, 1.8) and (15.5,1.8) .. (17, 1);
      \draw[-] (16, 1) .. controls (16.3, 1.26) and (16.7,1.26) .. (17, 1);

      \draw (-1.5, 0.5) node {$S[1..j]:$} (-1.5, 0.5);
      \draw (-1.2, -1.5) node {$G_{j}':$} (-1.2, -1.5);
      \draw (-1.2, -2.5) node {$G_{j}'':$} (-1.2, -2.5);
      \draw (-1.2, -3.5) node {$G_{j}:$} (-1.2, -3.5);

      \draw[rounded corners] (4,-2) rectangle (12.95,-1.1);
      \draw[rounded corners] (4,-3) rectangle (5,-2.1);
      \draw[rounded corners] (8,-3) rectangle (12.95,-2.1);
      \draw[rounded corners] (16,-3) rectangle (17,-2.1);
      \draw[rounded corners] (4,-4) rectangle (5,-3.1);
      \draw[rounded corners] (8,-4) rectangle (17,-3.1);

      \foreach \x in { 5, 9, 13 } \draw (\x-0.5,-2) node[above] {\x};
      \foreach \x in { 5, 9, 13, 17} \draw (\x-0.5,-3) node[above] {\x};
      \foreach \x in { 5, 9, 13, 17} \draw (\x-0.5,-4) node[above] {\x};

      \draw (4.5, -5) node[above] {$P_{j,\infty}$} (4.5, -5);
      \draw (12.5, -5) node[above] {$P_{j,4}$} (12.5, -5);

    \end{tikzpicture}
    \label{fig:G-update-example}
  }
  \caption{ (a) The palindromic suffixes of $S[1..j-1]$ for $j=17$
    start at positions $P_{j-1}=\{2,6,10,14,15,16\}$ and the compact
    representation is 
    $G_{j-1}=((2, \infty, 1), (6, 4,
    3), (15, 1, 2))$.  The shaded symbols will be compared with the next
    symbol appended to the text.\newline  (b) The palindromic suffixes 
    after appending $S[j]$. The sequence $G_j'$ is obtained by taking
    each triple $(i,\Delta,k)\in G_{j-1}$ and either removing it or
    replacing it with $(i-1,\Delta,k)$. The resulting sequence
    $G_j'=((5,4,3))$, however, is no longer a valid gap partitioning
    because the gap of the first element encoded by triple $(5,4,3)$
    is $\infty$. This is fixed by separating this element into its own
    triple. At this point we also add the palindromes of length at
    most 2 to obtain
    $G_{j}''=((5,\infty,1),(9,4,2),(17,4,1))$. Finally, we merge
    neighboring triples with the same $\Delta$ to obtain
    $G_j=((5,\infty,1),(9,4,3))$.  }\label{fig:G-example}
\end{figure}
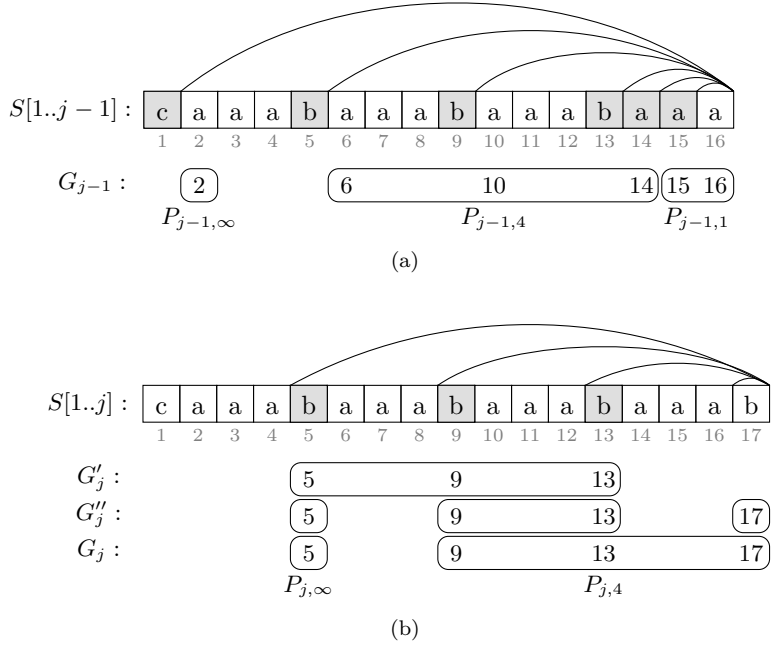

\begin{lemm}
\label{lemm:groups}
Let $p_i$ and $p_{i+1}$ be two consecutive elements of $P_{j-1,\Delta}$.
Then $p_i-1 \in P_{j}$ iff $p_{i+1}-1 \in P_{j}$.
\end{lemm}

\begin{proof}
By definition, $p_{i+1}-p_i=\Delta$, and the predecessor of $p_i$ in
$P_j$ is $p_{i-1}=p_i-\Delta$. Using the definitions from the proof
of Lemma~\ref{corr:number-of-gaps}, we have the situation of
Lemma~\ref{lemm:gap-properties}(3), which implies that 
$S[p_i-1]=S[p_{i+1}-1]=c$. Thus, $p_i-1\in P_j$ iff $S[j]=c$ iff
$p_{i+1}-1\in P_j$.
\end{proof}

Thus, when computing $G_j$, each triple $(i,\Delta,k)\in G_{j-1}$ will
be either eliminated or replaced by $(i-1,\Delta,k)$.
The resulting sequence of triples is
\[
G'_j = \{ (i-1,\Delta,k) : (i,\Delta,k) \in G_{j-1}, i>1, \text{ and }
S[i-1] = S[j] \}\,,
\]
which is a full representation of all palindromes longer than two in
$P_j$. 

However, the triples in $G'_j$ may no longer perfectly correspond to
the partitions $P_{j,\Delta}$ because the gaps may have
changed. Specifically, if the smallest element $p_i$ in
$P_{j-1,\Delta}$ is replaced by $p_i-1$ but its predecessor
$p_{i-1}=p_i-\Delta$ in $P_{j-1}$ is eliminated, then $p_i-1$ is not
in $P_{j,\Delta}$ but it is, at this point, represented by the triple
$(p_i-1,\Delta,k)$. Note that only the smallest element of each
partition can be affected by this. In such cases, we separate the
first element into its own triple, i.e., we replace $(p_i-1,\Delta,k)$
with $(p_i-1,\Delta',1)$ and (if $k>1$) $(p_i-1+\Delta,\Delta,k-1)$,
where $\Delta'$ is the new gap preceding $p_i-1$ in $P_j$. We will
also add separate triples to represent palindromes of lengths one and
(possibly) two.

Let $G''_j$ be the sequence of triples obtained from $G'_j$ by
the above process (see lines 8--21 in Figure~\ref{fig-algorithm}). 
It represents exactly the palindromes in $P_j$ and the $\Delta$-values
are now correct, but there may be multiple triples with the same
$\Delta$. Thus we obtain the final sequence $G_j$
from $G''_j$ by merging triples with the same $\Delta$.

The full procedure for computing $G_j$ from $G_{j-1}$ is shown on
lines~4--30 in Figure~\ref{fig-algorithm} and the example of computation
is given in Figure~\ref{fig:G-update-example}. Each triple is processed in
constant time and the number of triples never exceeds
$\Oh{|G_{j-1}|}$. 

\begin{lemm}
$G_{j}$ can be computed from $G_{j-1}$ in $\Oh{|G_{j-1}|}=\Oh{\log j}$ time.
\end{lemm}

\section{Faster Factorization}
\label{sec:factorization}

In this section, we will show how to compute $\PL[j]$ from
$\PL[0..j-1]$ and $G_j$ in $\Oh{|G_j|}$ time. The key to fast computation
of $G_j$ was the close relation between $P_{j,\Delta}$ and
$P_{j-1,\Delta}$. Now we will rely on the relation between
$P_{j,\Delta}$ and $P_{j-\Delta,\Delta}$ captured by the following
result.

\begin{lemm}
If $(i,\Delta,k)\in G_j$ for $k \geq 2$, then $(i,\Delta,k-1)\in
G_{j-\Delta}$.
\label{lemm:gpl-correctness}
\end{lemm}

\begin{proof}

By definition, $(i,\Delta,k)\in G_j$ is equivalent to saying that
$P_{j,\Delta}=\{i,i+\Delta,\dots,i+(k-1)\Delta\}$, and we need to show
that $P_{j-\Delta,\Delta}=\{i,i+\Delta,\dots,i+(k-2)\Delta\}$. We will
show first that
$P_{j-\Delta,\Delta}\cap[i-\Delta+1..j-\Delta]=\{i,i+\Delta,\dots,i+(k-2)\Delta\}$
and then that $P_{j-\Delta,\Delta}\cap[1..i-\Delta]=\emptyset$.

Since $y=S[i..j]$ and $x=S[i-\Delta..j]$ are palindromes and $y$ is
the longest proper border of $x$,
$S[i-\Delta..j-\Delta]=y=S[i..j]$. Thus for all $\ell\in[i..j]$,
$\ell\in P_j$ iff $\ell-\Delta\in P_{j-\Delta}$ (see
Figure~\ref{fig:gpl-correctness-1}). In particular, the
gaps in both cases are the same and for all
$\ell\in[i+1..j]$, $\ell\in P_{j,\Delta}$ iff $\ell-\Delta\in
P_{j-\Delta,\Delta}$. Thus $P_{j-\Delta,\Delta}\cap[i-\Delta+1..j-\Delta] =
\{i,i+\Delta,\dots,i+(k-2)\Delta\}$.

We still need to show that
$P_{j-\Delta,\Delta}\cap[1..i-\Delta]=\emptyset$, which is true if and
only if $i-2\Delta\not\in P_{j-\Delta}$. Suppose to the contrary that
$S[i-2\Delta..j-\Delta]$ is a palindrome and let
$w=S[i-2\Delta..i-\Delta-1]$. Then $S[j-2\Delta+1..j-\Delta]=\rev{w}$,
the reverse of $w$.
Since $z=S[i-\Delta..j-\Delta]$ and $S[i-\Delta..j]$ are palindromes
too, we have that $S[i-\Delta..i-1]=w$ and $S[j-\Delta+1..j]=\rev{w}$.
Finally, since $z$ is a palindrome, $S[i-2\Delta..j]=wz\rev{w}$ is a
palindrome (see Figure~\ref{fig:gpl-correctness-2}). This implies that
$i-2\Delta\in P_{j}$ and thus $i-\Delta\in P_{j,\Delta}$, which
is a contradiction.
\end{proof}

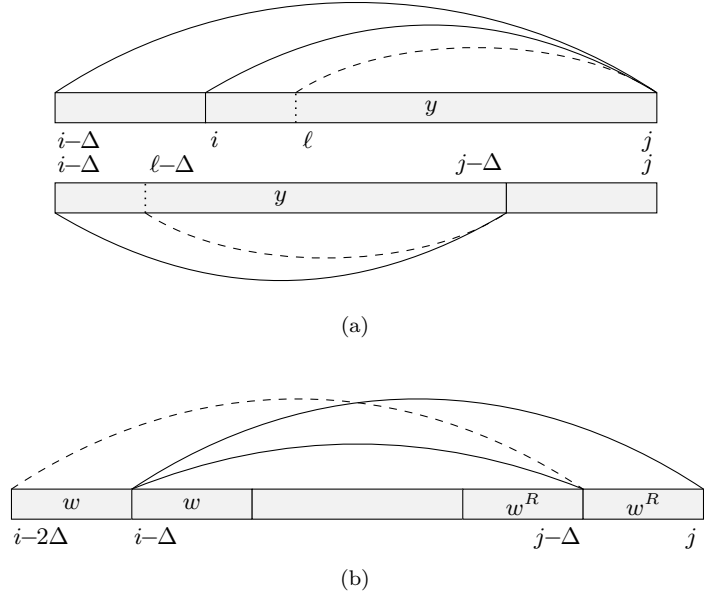
\begin{figure}[t]
  \centering
  \subfloat[]{
  \centering\small
  \begin{tikzpicture}[scale=0.40]
    \draw[fill=gray!10] (0, 2) rectangle (5, 3);
    \draw[fill=gray!10] (5, 2) rectangle (20, 3);
    \draw[fill=gray!10] (0, -1) rectangle (15, 0);
    \draw[fill=gray!10] (15, -1) rectangle (20, 0);

    \draw (12.5, 2.5) node {$y$} rectangle (12.5, 2.5);
    \draw (0.8, 2) node[below] {$i\hspace{-0.1cm}-\hspace{-0.1cm}\Delta$} rectangle (0.8, 2);
    \draw (0.8, 0) node[above] {$i\hspace{-0.1cm}-\hspace{-0.1cm}\Delta$} rectangle (0.8, 0);
    \draw (14.1, 0) node[above] {$j\hspace{-0.1cm}-\hspace{-0.1cm}\Delta$} rectangle (14.1, 0);
    \draw (3.9, 0) node[above] {$\ell\hspace{-0.1cm}-\hspace{-0.1cm}\Delta$} rectangle (3.9, 0);
    \draw (5.3, 2) node[below] {$i$} rectangle (5.3, 2);
    \draw (19.7, 2) node[below] {$j$} rectangle (19.7, 2);
    \draw (19.7, 0) node[above] {$j$} rectangle (19.7, 0);
    \draw (8.4, 2) node[below] {$\ell$} rectangle (8.4, 2);
    \draw (7.5, -0.5) node {$y$} rectangle (7.5, -0.5);

    \draw[-, dotted, semithick] (8.0, 2) -- (8.0, 3);
    \draw[-, dotted, semithick] (3.0, -1) -- (3.0, 0);

    \draw[-] (0, 3) .. controls (6,7) and (14,7) .. (20, 3);
    \draw[-] (5, 3) .. controls (10,6) and (15,6) .. (20, 3);
    \draw[-, dashed] (8, 3) .. controls (11,5) and (17,5) .. (20, 3);
    \draw[-, dashed] (3, -1) .. controls (6, -3) and (12, -3) .. (15, -1);
    \draw[-] (0, -1) .. controls (5, -4) and (10, -4) .. (15, -1);
  \end{tikzpicture}
  \label{fig:gpl-correctness-1}
  }

  \subfloat[]{
  \centering\small
  \begin{tikzpicture}[scale=0.40]
    \draw[fill=gray!10] (0,0) rectangle (4,1);
    \draw[fill=gray!10] (4,0) rectangle (8,1);
    \draw[fill=gray!10] (8,0) rectangle (15,1);
    \draw[fill=gray!10] (15,0) rectangle (19,1);
    \draw[fill=gray!10] (19,0) rectangle (23,1);

    \draw (2,0.5) node {$w$} rectangle (2,0.5);
    \draw (6,0.5) node {$w$} rectangle (6,0.5);
    \draw (17,0.5) node {$w^{R}$} rectangle (17,0.5);

    \draw (21,0.5) node {$w^{R}$} rectangle (21,0.5);

    \draw (22.6,0) node[below] {$j$} rectangle (22.6,0);
    \draw (18.15,0) node[below] {$j\hspace{-0.1cm}-\hspace{-0.1cm}\Delta$} rectangle (18.15,0);
    \draw (4.8,0) node[below] {$i\hspace{-0.1cm}-\hspace{-0.1cm}\Delta$} rectangle (4.8,0);
    \draw (1,0) node[below] {$i\hspace{-0.1cm}-\hspace{-0.1cm}2\Delta$} rectangle (1,0);

    \draw[-] (4, 1) .. controls (9,3) and (14,3) .. (19, 1);
    \draw[-] (4, 1) .. controls (10,5) and (17,5) .. (23, 1);
    \draw[-,dashed] (0, 1) .. controls (6,5) and (13,5) .. (19, 1);
  \end{tikzpicture}
  \label{fig:gpl-correctness-2}
  }
  \caption{Proof of Lemma~\ref{lemm:gpl-correctness}. (a) $\ell\in P_j$ iff
    $\ell-\Delta\in P_{j-\Delta}$ for all $\ell\in[i..j]$. (b)
    If $i-2\Delta\in P_{j-\Delta}$
    then $S[i-2\Delta..j]$ is a palindrome.}
\end{figure}

\begin{figure}[th!]
  \centering
  \subfloat[Iteration $j$.]{
    \hspace{-1.5cm}
    \small
    \begin{tikzpicture}[scale=0.49]
    \newcommand{\minbox}[2]{\draw (#1, #2) ellipse (0.8cm and 0.5cm);\draw (#1, #2+0.05) node {$\min$};}

    \foreach \x/\ch in { 1/c, 2/a, 3/a, 4/a, 5/b, 6/a, 7/a, 8/a, 9/b, 10/a, 11/a, 12/a, 13/b, 14/a, 15/a, 16/a}
      \draw (\x-1,0) rectangle (\x,1);
    \foreach \x/\ch in { 1/c, 2/a, 3/a, 4/a, 5/b, 6/a, 7/a, 8/a, 9/b, 10/a, 11/a, 12/a, 13/b, 14/a, 15/a, 16/a}
      \draw (\x-0.5,0) node[above] {\normalsize {\ch}};
    \foreach \x in {0, 1, 2, 3, 4, 5, 6, 7, 8, 9, 10, 11, 12, 13, 14, 15, 16}
      \draw (\x-0.5,-3.25) node[above] {\scriptsize \textcolor{gray}{\x}};

    \draw[-] (1, 1) .. controls (5,4.2) and (12,4.2) .. (16, 1);
    \draw[-] (5, 1) .. controls (8,3.2) and (13,3.2) .. (16, 1);
    \draw[-] (9, 1) .. controls (10.5,2.4) and (14.5,2.4) .. (16, 1);
    \draw[-] (13, 1) .. controls (14, 1.8) and (15,1.8) .. (16, 1);
    \draw[-] (14, 1) .. controls (14.7, 1.5) and (15.3,1.5) .. (16, 1);
    \draw[-] (15, 1) .. controls (15.3, 1.3) and (15.7,1.3) .. (16, 1);

    \draw (-1, 0.5) node[left] {$S[1..j]:$} (-1, 0.5);
    \draw (-1, -0.5) node[left] {$G_{j}:$} (-1, -0.5);
    \draw (-1, -3.75) node[left] {$\PL:$} (-1, -3.75);

    \draw[rounded corners] (1,-1) rectangle (2,-0.1);
    \draw[rounded corners] (5,-1) rectangle (13.95,-0.1);
    \draw[rounded corners] (14.05,-1) rectangle (16,-0.1);

    \foreach \x in { 2, 6, 10, 14, 15, 16} \draw (\x-0.5,-1) node[above] {\x};
    \draw (1.5, -2) node[above] {$P_{j,\infty}$};
    \draw (9.5, -2) node[above] {$P_{j,4}$};
    \draw (15, -2) node[above] {$P_{j,1}$};

    \foreach \x in {0, 1, 2, 3, 4, 5, 6, 7, 8, 9, 10, 11, 12, 13, 14, 15, 16}
      \draw (\x-1,-4.25) rectangle (\x,-3.25);
    \foreach \x/\ch in { 5, 9, 13}
      \draw[fill=gray!30] (\x-1, -4.25) rectangle (\x,-3.25);
    \foreach \x/\ch in {0/0, 1/1, 2/2, 3/2, 4/2, 5/3, 6/3, 7/3, 8/2, 9/3, 10/3, 11/3, 12/2, 13/3, 14/3, 15/3, 16/2}
      \draw (\x-0.5,-4.25) node[above] {\normalsize {\ch}};

    \draw[->] (4.3,-4.25) -- (4.3,-6);
    \draw (4.4,-4.65) node[left] {\scriptsize +1};
    \draw (8.5,-4.65) node[left] {\scriptsize +1};
    \draw (12.5,-4.65) node[left] {\scriptsize +1};
    \draw[->, rounded corners=0.6mm] (8.5, -4.25) -- (8.5, -5) -- (4.5, -5) -- (4.5, -6);

    \draw[->, rounded corners=0.6mm] (12.5, -4.25) -- (12.5, -5.2) -- (4.7, -5.2) -- (4.7, -6);

    \minbox{4.5}{-6.5};

    \draw[->] (13.4,-4.25) -- (13.4,-6);
    \draw (13.55,-4.65) node[left] {\scriptsize $+1$};
    \draw (13.5,-4.65) node[right] {\scriptsize +1};
    \draw[->, rounded corners=0.6mm] (14.5, -4.25) -- (14.5, -5) -- (13.6, -5) -- (13.6, -6);
    \minbox{13.5}{-6.5};

    \draw[->,rounded corners=0.6mm] (0.2,-4.25) -- (0.2,-10.7) -- (15.7, -10.7) -- (15.7, -6);
    \draw (0.3, -4.65) node[left] {\scriptsize {+1}};

    \minbox{15.5}{-5.5};
    \draw[->] (15.5,-5) -- (15.5,-4.25);

    \foreach \x in {2, 14}
      \draw (\x-1,-9) rectangle (\x,-8);
    \draw[-] (0.5, -8) -- (2.5, -8);
    \draw[-] (0.5, -9) -- (2.5, -9);
    \draw[-] (12.5, -8) -- (14.5, -8);
    \draw[-] (12.5, -9) -- (14.5, -9);
    \foreach \x/\ch in { 1.1, 2.9, 13.1, 14.9}
      \draw (\x-0.5,-8.75) node[above] {\normalsize {..}};

    \foreach \x/\ch in { 2/4, 14/4 }
      \draw (\x-0.5,-9) node[above] {\normalsize {\ch}};

    \foreach \x in {2, 14}
      \draw (\x-1,-9) rectangle (\x,-8);

    \draw[->] (13.2,-7) -- (13.2,-8);
    \draw[->, rounded corners=0.6mm] (4.5,-7) -- (4.5, -7.5) -- (1.8, -7.5) -- (1.8,-8);
    \draw (13.5, -8) node[above] {\scriptsize \textcolor{gray}{14}};
    \draw (1.5, -8) node[above] {\scriptsize \textcolor{gray}{2}};

    \draw (1.6, -9) node[below] {$\PL_{j,4}$};
    \draw (13.5, -9) node[below] {$\PL_{j,1}$};

    \fill[fill=white] (0, -7) rectangle (1, -6);
    \draw (0.2, -6.5) node {$\PL_{j,\infty}=2$};

    \draw (-1, -8.5) node[left] {$\GPL:$} (-1, -8.5);
    \draw[->, rounded corners=0.6mm] (13.5,-10) -- (13.5, -10.3) -- (15.3, -10.3) -- (15.3,-6);
    \draw[->, rounded corners=0.6mm] (1.5,-10) -- (1.5, -10.5) -- (15.5, -10.5) -- (15.5,-6);
  \end{tikzpicture}
  }

  \subfloat[Iteration $j-4$.]{
    \hspace{-2.0cm}
    \small
    \begin{tikzpicture}[scale=0.49]
    \newcommand{\minbox}[2]{\draw (#1, #2) ellipse (0.8cm and 0.5cm);\draw (#1, #2+0.05) node {$\min$};}

    \foreach \x/\ch in { 1/c, 2/a, 3/a, 4/a, 5/b, 6/a, 7/a, 8/a, 9/b, 10/a, 11/a, 12/a}
      \draw (\x-1,0) rectangle (\x,1);
    \foreach \x/\ch in { 1/c, 2/a, 3/a, 4/a, 5/b, 6/a, 7/a, 8/a, 9/b, 10/a, 11/a, 12/a}
      \draw (\x-0.5,0) node[above] {\normalsize {\ch}};
    \foreach \x in {0, 1, 2, 3, 4, 5, 6, 7, 8, 9, 10, 11, 12}
      \draw (\x-0.5,-3.25) node[above] {\scriptsize \textcolor{gray}{\x}};

    \draw[-] (1, 1) .. controls (4,3.2) and (9,3.2) .. (12, 1);
    \draw[-] (5, 1) .. controls (6.5,2.4) and (10.5,2.4) .. (12, 1);
    \draw[-] (9, 1) .. controls (10, 1.8) and (11,1.8) .. (12, 1);
    \draw[-] (10, 1) .. controls (10.7, 1.5) and (11.3,1.5) .. (12, 1);
    \draw[-] (11, 1) .. controls (11.3, 1.3) and (11.7,1.3) .. (12, 1);

    \draw (-1, 0.5) node[left] {$S[1..j-4]:$};
    \draw (-1, -0.5) node[left] {$G_{j-4}:$};
    \draw (-1, -3.75) node[left] {$\PL:$};

    \draw[rounded corners] (1,-1) rectangle (2,-0.1);
    \draw[rounded corners] (5,-1) rectangle (9.95,-0.1);
    \draw[rounded corners] (10.05,-1) rectangle (12,-0.1);

    \foreach \x in { 2, 6, 10, 11, 12} \draw (\x-0.5,-1) node[above] {\x};
    \draw (1.5, -2) node[above] {$P_{j-4,\infty}$};
    \draw (7.5, -2) node[above] {$P_{j-4,4}$};
    \draw (11, -2) node[above] {$P_{j-4,1}$};

    \foreach \x/\ch in { 5, 9 }
      \draw[fill=gray!30] (\x-1, -4.25) rectangle (\x,-3.25);
    \foreach \x in {0, 1, 2, 3, 4, 5, 6, 7, 8, 9, 10, 11, 12}
      \draw (\x-1,-4.25) rectangle (\x,-3.25);
    \foreach \x/\ch in {0/0, 1/1, 2/2, 3/2, 4/2, 5/3, 6/3, 7/3, 8/2, 9/3, 10/3, 11/3, 12/2}
      \draw (\x-0.5,-4.25) node[above] {\normalsize {\ch}};

    \draw[->] (4.4,-4.25) -- (4.4,-6);
    \draw (4.5,-4.65) node[left] {\scriptsize +1};
    \draw (8.5,-4.65) node[left] {\scriptsize +1};
    \draw[->, rounded corners=0.6mm] (8.5, -4.25) -- (8.5, -5) -- (4.6, -5) -- (4.6, -6);
    \minbox{4.5}{-6.5};

    \draw[->] (9.4,-4.25) -- (9.4,-6);
    \draw (9.55,-4.65) node[left] {\scriptsize $+1$};
    \draw (9.5,-4.65) node[right] {\scriptsize +1};
    \draw[->, rounded corners=0.6mm] (10.5, -4.25) -- (10.5, -5) -- (9.6, -5) -- (9.6, -6);
    \minbox{9.5}{-6.5};

    \draw[->,rounded corners=0.6mm] (0.2,-4.25) -- (0.2,-10.7) -- (11.7, -10.7) -- (11.7, -6);
    \draw (0.3, -4.65) node[left] {\scriptsize {+1}};

    \minbox{11.5}{-5.5};
    \draw[->] (11.5,-5) -- (11.5,-4.25);

    \draw[-] (0.5, -8) -- (2.5, -8);
    \draw[-] (0.5, -9) -- (2.5, -9);
    \draw[-] (8.5, -8) -- (10.5, -8);
    \draw[-] (8.5, -9) -- (10.5, -9);
    \foreach \x/\ch in { 1.1, 2.9, 9.1, 10.9}
      \draw (\x-0.5,-8.75) node[above] {\normalsize {..}};

    \foreach \x/\ch in { 2/4, 10/4 }
      \draw (\x-0.5,-9) node[above] {\normalsize {\ch}};

    \foreach \x in {2, 10}
      \draw (\x-1,-9) rectangle (\x,-8);

    \draw[->] (9.2,-7) -- (9.2,-8);
    \draw[->, rounded corners=0.6mm] (4.5,-7) -- (4.5, -7.5) -- (1.8, -7.5) -- (1.8,-8);
    \draw (9.5, -8) node[above] {\scriptsize \textcolor{gray}{10}};
    \draw (1.5, -8) node[above] {\scriptsize \textcolor{gray}{2}};

    \draw (1.6, -9) node[below] {$\PL_{j-4,4}$};
    \draw (9.5, -9) node[below] {$\PL_{j-4,1}$};

    \fill[fill=white] (0, -7) rectangle (1, -6);
    \draw (0.2, -6.5) node {$\PL_{j-4,\infty}=2$};

    \draw (-1, -8.5) node[left] {$\GPL:$};
    \draw[->, rounded corners=0.6mm] (9.5,-10) -- (9.5, -10.3) -- (11.3, -10.3) -- (11.3,-6);
    \draw[->, rounded corners=0.6mm] (1.5,-10) -- (1.5, -10.5) -- (11.5, -10.5) -- (11.5,-6);
   \end{tikzpicture}
  }
   \caption{Example usage of the $\GPL$ array for $j=16$.
     The value of $\PL_{j,4}$ computed in iteration $j$ depends on shaded
     elements from $\PL$ array. Rather than scanning them all, we apply
     Lemma~\ref{lemm:gpl-correctness}. Since $|P_{j,4}|\geq 2$ we get
     $P_{j,4}=P_{j-4,4} \cup \{14\}$. Therefore we can compute $\PL_{j,4}$ as
     $\min\{\PL_{j-4,4},\PL[13]+1\}$. The value of $\PL_{j-4,4}$ was computed during
     iteration $j-4$ and stored at position $\min P_{j-4,4}-4 = \min P_{j,4}-4=2$ in the $\GPL$
     array, and by Lemma~\ref{lemm:no-overwriting-lemma} it was not overwritten
     between iterations $j-4$ and $j$. Thus we compute $\PL_{j,4}$ in constant
     time as $\min\{\GPL[2],\PL[13]+1\}$ and update $\GPL[2]$ with the new value.
   }\label{fig:GPL-example}
\end{figure}

By the above lemma, $P_{j,\Delta}=P_{j-\Delta,\Delta}\cup\{\max
P_{j,\Delta}\}$ whenever $|P_{j,\Delta}|\ge 2$. Thus we can compute
$\PL_{j,\Delta} = \min\{\PL[i-1]+1 : i\in P_{j,\Delta}\}$ from
$\PL_{j-\Delta,\Delta}$ in constant time. We will store the value
$\PL_{j,\Delta}$ in an array $\GPL[1..n]$ at the position $m=\min
P_{j,\Delta}-\Delta$. Note that $m$ is the predecessor of $\min
P_{j,\Delta}$ in $P_j$ and the position is shared by
$\PL_{j-\Delta,\Delta}$ (when $|P_{j,\Delta}| \ge 2$).  The following
lemma shows that the position is not overwritten by another value
between the rounds $j-\Delta$ and $j$. See Figure~\ref{fig:GPL-example}
for an example.

\begin{lemm}
Let $m=\min P_{j,\Delta}-\Delta$. For all $\ell\in[j-\Delta+1..j-1]$, 
$m \not\in P_\ell$.
\label{lemm:no-overwriting-lemma}
\end{lemm}

\begin{proof}
  Suppose to the contrary that $m\in P_\ell$ for some
  $\ell\in[j-\Delta+1..j-1]$, i.e., $S[m..\ell]$ is a palindrome.
  Then $S[m+h..\ell-h]$ for $h=\ell-j+\Delta$ is a palindrome too (see
  Figure~\ref{fig:no-overwriting-lemma}).
  Since $\ell-h=j-\Delta$ and $m<m+h<m+\Delta=\min
  P_{j-\Delta,\Delta}$, this contradicts $m$ being the predecessor of
  $\min P_{j-\Delta,\Delta}$ in $P_{j-\Delta}$.
\end{proof}

\begin{figure}[t]
  \centering\small
  \begin{tikzpicture}[scale=0.40]

  \draw[fill=gray!10] (0,0) rectangle (4,1);
  \draw[fill=gray!10] (4,0) rectangle (6,1);
  \draw[fill=gray!10] (6,0) rectangle (14,1);
  \draw[fill=gray!10] (14,0) rectangle (18,1);
  \draw[fill=gray!10] (18,0) rectangle (20,1);

  \draw (17.6,0) node[below] {$\ell$} rectangle (17.6,0);
  \draw (19.8,0) node[below] {$j$} rectangle (19.8,0);
  \draw (13.1,0) node[below] {$j\hspace{-0.1cm}-\hspace{-0.1cm}\Delta$} rectangle (13.1,0);
  \draw (7.1,0) node[below] {$m\hspace{-0.1cm}+\hspace{-0.1cm}\Delta$} rectangle (7.1,0);
  \draw (0.4,0) node[below] {$m$} rectangle (0.4,0);

  \draw[-] (0, 1) .. controls (5,5) and (13,5) .. (18, 1);
  \draw[-,dashed] (4, 1) .. controls (7,2.5) and (11,2.5) .. (14, 1);
  
  \draw[<->] (0,1.5) -- (4,1.5);
  \draw[<->] (14,1.5) -- (18,1.5);

  \draw (2, 1.5) node[above] {$h$} rectangle (2, 1.5);
  \draw (16, 1.5) node[above] {$h$} rectangle (16, 1.5);

  \draw[dotted] (0,1.5) -- (0,1);
  \draw[dotted] (4,1.5) -- (4,1);
  \draw[dotted] (14,1.5) -- (14,1);
  \draw[dotted] (18,1.5) -- (18,1);

  \end{tikzpicture}
  \caption{Proof of Lemma~\ref{lemm:no-overwriting-lemma}: if $m\in P_{\ell}$
    then $m+h\in P_{j-\Delta}$.}
\label{fig:no-overwriting-lemma}
\end{figure}
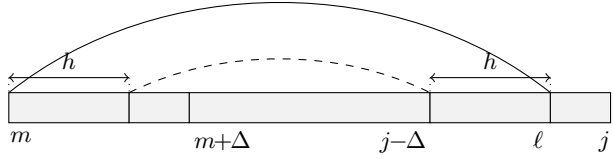

\begin{figure}[ht!]
\centering
\begin{tabbing}
00: \=\qquad\=\qquad\=\qquad\=\qquad\=\qquad\=\kill
 \balgorithm{Palindromic-length}($S[1..n]$)\\
 1:\>$\PL[0] \la 0$\\
 2:\>$G \la ()$\\
 3:\>\bfor $j \la 1$ \bto $n$ \bdo\\
 4:\>\>$G' \la ()$\\
 5:\>\>\bforeach $(i,\Delta,k)\in G$ \bdo\\
 6:\>\>\>\bif $i>1$ \band $S[i-1]=S[j]$ \bthen\\
 7:\>\>\>\>$G'.\text{pushback}((i-1,\Delta,k))$ \qquad // appends the given triple\\
 8:\>\>$G'' \la ()$\\
 9:\>\>$r \la -j$\qquad// makes $i-r$ big enough to act as $\infty$\\
 10:\>\>\bforeach $(i,\Delta,k)\in G'$ \bdo\\
 11:\>\>\>\bif $i-r \neq \Delta$ \bthen\\
 12:\>\>\>\>$G''.\text{pushback}((i,i-r,1))$\\
 13:\>\>\>\>\bif $k>1$ \bthen\\
 14:\>\>\>\>\>$G''.\text{pushback}((i+\Delta,\Delta,k-1))$\\
 15:\>\>\>\belse\\
 16:\>\>\>\>$G''.\text{pushback}((i,\Delta,k))$\\
 17:\>\>\>$r \la i+(k-1)\Delta$\\
 18:\>\>\bif $j > 1$ \band $S[j-1]=S[j]$ \bthen\\
 19:\>\>\>$G''.\text{pushback}((j-1,j-1-r,1))$\\
 20:\>\>\>$r \la j-1$\\
 21:\>\>$G''.\text{pushback}((j,j-r,1))$\\
 22:\>\>$G \la ()$\\
 23:\>\>$(i',\Delta',k') \la G''.\text{popfront}()$ \qquad // removes and returns the first triple\\
 24:\>\>\bforeach $(i,\Delta,k)\in G''$ \bdo\\
 25:\>\>\>\bif $\Delta' = \Delta$ \bthen\\
 26:\>\>\>\>$k'=k'+k$\\
 27:\>\>\>\belse\\
 28:\>\>\>\>$G.\text{pushback}((i',\Delta',k'))$\\
 29:\>\>\>\>$(i',\Delta',k') \la (i,\Delta,k)$\\
 30:\>\>$G.\text{pushback}((i',\Delta',k'))$\\
 31:\>\>$\PL[j] \la j$\\
 32:\>\>\bforeach $(i,\Delta,k)\in G$ \bdo\\
 33:\>\>\>$r \la i + (k-1)\Delta$\\
 34:\>\>\>$m \la \PL[r-1]+1$\\
 35:\>\>\>\bif $k>1$ \bthen\\
 36:\>\>\>\>$m \la \min(m,\GPL[i-\Delta])$\\
 37:\>\>\>\bif $\Delta \leq i$ \bthen\\
 38:\>\>\>\>$\GPL[i-\Delta] \la m$\\
 39:\>\>\>$\PL[j] \la \min(\PL[j],m)$\\
 40:\>\breturn $\PL[n]$
\end{tabbing}
\caption{Algorithm for computing the palindromic length in $\Oh{n\log
    n}$ time.}
\label{fig-algorithm}
\end{figure}

The full algorithm is given in Figure~\ref{fig-algorithm}. The running
time of round $j$ is $\Oh{|G_{j-1}|+|G_j|}$. Since $|G_j|=\Oh{\log j}$
for all $j$, we obtain the following result.

\begin{thrm}
  The palindromic length of a string of length $n$ can be computed in
  $\Oh{n\log n}$ time and $\Oh{n}$ space.
\end{thrm}

As with the quadratic-time algorithm, the algorithm can be modified to
produce an actual minimum palindromic factorization without an
asymptotic increase in time or space complexities: we need only store
with each palindromic length in $\PL$ and $\GPL$, the length of the
last palindrome in the corresponding minimum factorization.  The
algorithm is also online in the sense that the string is processed
from left to right and, for each $j$, the character $S[j]$ is
processed in $\Oh{\log j}$ time, after which we can report the
palindromic length $\PL(S[1..j])$ in constant time and the
corresponding factorization in $\Oh{\PL(S[1..j])}$ time.

\section{Average and Worst Case}

In this section, we show that the average case time complexity of the
algorithm is linear, but that the worst case is indeed
$\Theta(n\log n)$.

\begin{thrm}
\label{thrm:average}
  The average case time complexity of the algorithms in
  Figure~\ref{fig-quadratic-algorithm} and in
  Figure~\ref{fig-algorithm} is $\Oh{n}$.
\end{thrm}

\begin{proof}
  Consider the set $\Sigma^n$ of the $\sigma^n$ strings of length $n$ over
  an alphabet $\Sigma$ of size $\sigma>1$. All of them have a
  palindromic suffix of length one, $\sigma^{n-1}$ of them have a
  palindromic suffix of length two, and the same number have a
  palindromic suffix of length three (assuming $n\ge 3$). More
  generally, for $1 \le k \le n$, the number of strings with a
  palindromic suffix of length $k$ is $\sigma^{n-k/2}$ when $k$ is
  even and $\sigma^{n-(k-1)/2}$ when $k$ is odd. Then the total number
  of palindromic suffixes in $\Sigma^n$ is
  \[
  \sum_{i=1}^{\lfloor n/2\rfloor} \sigma^{n-i} + 
  \sum_{i=1}^{\lceil n/2\rceil} \sigma^{n-i+1} < \sigma^{n}/(\sigma-1)
  + \sigma^{n+1}/(\sigma-1) = \frac{\sigma+1}{\sigma-1}\sigma^n \le
  3\sigma^n\ .
  \]
  Therefore the average number of palindromes ending at any position is
  less than three, and both algorithms spend a constant time on
  average for processing each position.
\end{proof}

We show the worst case complexity of the algorithm by constructing a
family of strings based on the Zimin words~\cite[Chapter~5.4]{BLRS08}.  Let
$Z_0=\varepsilon$, and $Z_i=Z_{i-1}iZ_{i-1}$ for $i>0$. The limit of
this sequence is the infinite Zimin word
$Z=1213121412131215\dots$. For a non-negative integer $n$, let $B(n)$
be the number of 1-bits in the binary representation of $n$. For
example, $B(0)=0$, $B(1)=1$, $B(7)=3$ and $B(8)=1$.

\begin{lemm}
\label{lm:zimin}
  The prefix $Z[1..n]$ of the infinite Zimin word $Z$ has exactly
  $B(n)$ suffix palindromes.
\end{lemm}

\begin{proof}
  From the definition, it is easy to see that
  the prefix $Z[1..n]$ has a unique factorization of the form
  \[
  Z[1..n] = Z_{i_k}(i_{k}+1) \cdot Z_{i_{k-1}}(i_{k-1}+1) \cdots Z_{i_2}(i_2+1)
  \cdot Z_{i_1}(i_1+1)
  \]
  where $0 \le i_1 < i_2 < \dots < i_{k-1} < i_k$. For example,
  $Z[1..10]=Z_34Z_12$.
  Since the length of
  a factor $Z_{i}(i+1)$ is $2^i$, we must have that $\sum_{j=1}^k
  2^{i_j}=n$. Thus $i_1,\dots,i_k$ are the positions of 1-bits in the
  binary representation of $n$, and $k=B(n)$.

  Let $n_j=2^{i_j}$ for $j\in[1..k]$.  Clearly, $Z[2n_k-n..n]$ is a
  palindrome of length $2(n-n_k)+1$ centered at $Z[n_k]=(i_k+1)$. For
  example, $Z[6..10]=21412$ is a palindrome centered at $Z[8]=4$.
  Since $Z[n_k]$ is the only occurrence of $(i_k+1)$ in $Z[1..n]$,
  there can be no other suffix palindromes with a starting position
  in $Z[1..n_k]$. By a similar argument, there is exactly one suffix
  palindrome with a starting position in $Z[n_k+1..n_k+n_{k-1}]$, the one
  centered at $Z[n_k+n_{k-1}]=(i_{k-1}+1)$, and so on. In total,
  $Z[1..n]$ has exactly $k$ suffix palindromes.
\end{proof}

\begin{thrm}
  The running time of the algorithm in Figure~\ref{fig-algorithm} for
  input $Z[1..n]$ is $\Theta(n\log n)$.
\end{thrm}

\begin{proof}
  By Lemma~\ref{lm:zimin}, $Z[1..j]$ has exactly $B(j)$ suffix
  palindromes, i.e., $|P_j|=B(j)$. From the proof it is easy to see that
  each of the suffix palindromes is at least twice as long as the next
  shorter suffix palindrome. Thus there are no two identical gaps in
  $P_j$ and $|G_j|=|P_j|=B(j)$. Since the algorithm spends
  $\Theta(|G_{j-1}|+|G_j|)$ time in round $j$, the total time complexity is
  $\Theta\!\left(\sum_{j=1}^n B(j)\right)$, which is $\Theta(n\log
  n)$~\cite{McI74}. 
\end{proof}

\section*{Acknowledgements}

Many thanks to the organizers and participants of the Stringmasters
2013 workshops in Verona and Prague, and to the anonymous reviewers.
This research was partially supported by the Italian
MIUR Project PRIN 2010LYA9RH, ``Automi e Linguaggi Formali: Aspetti
Matematici e Applicativi'', and by the Academy of Finland
through grant 268324 and grant 118653 (ALGODAN).


\end{document}